\documentclass[conference]{IEEEtran}
\IEEEoverridecommandlockouts
\usepackage{cite,amsmath,amssymb,amsfonts,amsthm,graphicx,textcomp,xcolor,booktabs,bm,mathtools,url,caption}
\usepackage[hidelinks]{hyperref}
\usepackage{orcidlink}
\usepackage{tikz}

\newtheorem{theorem}{Theorem}
\newtheorem{proposition}[theorem]{Proposition}

\newtheorem{assumption}{Assumption}
\newtheorem{remark}{Remark}

\DeclareMathOperator*{\argmin}{arg\,min}

\begin{document}
\title{Composite Adaptive Higher-Order Control Barrier Functions for Joint Frequency--RoCoF Safety in Low-Inertia SIDS Microgrids}

\author{\IEEEauthorblockN{Shankar Ramharack\,\orcidlink{0000-0003-3759-7333}}
\IEEEauthorblockA{\textit{Engineering Department} \\
\textit{M\&E Partners (Barbados) Limited} \\
St.~Philip, Barbados, W.I.\\
sramharack@ieee.org}
\and
\IEEEauthorblockN{Jonathan Nancoo\,}
\IEEEauthorblockA{\textit{Department of Electrical and Computer Engineering} \\
\textit{The University of the West Indies}\\
St.~Augustine, Trinidad and Tobago\\
jonathan.nancoo@uwi.edu}
\and
\IEEEauthorblockN{Rajiv Sahadeo}
\IEEEauthorblockA{Trinidad and Tobago, W.I.\\
rajivsahadeo@outlook.com}
}

\maketitle

\let\thefootnote\relax\footnotetext{Code and data: \url{https://github.com/sramharack/caribcon26-joint-safety-vsg}. Correspondence: sramharack@ieee.org.}

\begin{abstract}
Small Island Developing States (SIDS) face simultaneous frequency nadir and rate-of-change-of-frequency (RoCoF) violations under high renewable penetration. This paper proposes a \emph{Composite Adaptive Higher-Order Control Barrier Function} (CA-HOCBF) for BESS-coupled virtual synchronous generators that jointly enforces IEEE~1547 frequency ($\pm0.8$\,Hz) and ENTSO-E RoCoF ($\pm1$\,Hz/s) limits as hard safety constraints under parametric uncertainty. The method unifies three tools: (i)~a dual barrier architecture with a zeroing CBF for frequency and an algebraic RoCoF constraint with a robust inertia-floor fallback; (ii)~a composite energy function coupling logarithmic safety barriers with quadratic parameter and disturbance error terms, which drives safety-weighted adaptation; and (iii)~a disturbance observer integrated into the barrier dynamics. We prove that the closed-form QP is always feasible under a BESS capacity condition and, under explicitly stated assumptions, that the composite energy remains finite along trajectories, rendering the joint safe set forward invariant \emph{without requiring parameter convergence}. Simulation on a 10\,MW Caribbean microgrid ($H=2$\,s, 70\% renewables) shows the CA-HOCBF eliminates all frequency and RoCoF violations under a compound disturbance, achieving steady-state error of 0.001\,Hz and RoCoF of 0.50\,Hz/s, where fixed-gain VSGs suffer a 0.62\,Hz offset with 2.25\,Hz/s RoCoF violations.
\end{abstract}
 

\section{Introduction}

\subsection{Motivation}

SIDS are pursuing aggressive renewable targets ($>70\%$) to reduce fossil fuel dependence and meet climate commitments~\cite{irena2024sids}. The displacement of synchronous generators by inverter-based resources (IBRs) critically reduces rotational inertia, making frequency stability the binding operational constraint in these isolated grids~\cite{he2024frequency, milano2018foundations}. Under a single contingency, SIDS microgrids can experience RoCoF exceeding 2\,Hz/s and frequency nadirs violating IEEE~1547-2018~\cite{ieee1547} disconnection thresholds, triggering cascading under-frequency load shedding~\cite{hatziargyriou2021, kundur2004}.

\subsection{Related Work}
\subsubsection{Virtual synchronous generators} 
BESS-coupled VSGs emulate synchronous machine dynamics to provide synthetic inertia and primary frequency response~\cite{bevrani2014, lasseter2020gfm, khan2024gfm}. Zhong and Weiss~\cite{zhong2011synchronverter} proposed the synchronverter, and D'Arco et al.~\cite{darco2015vsg} established virtual inertia control architectures. However, fixed-parameter designs cannot track the dramatic inertia variations in SIDS as IBR dispatch changes hourly~\cite{markovic2021, pattabiraman2018}. Adaptive VSGs have been explored by Shi et al.~\cite{shi2018adaptive} and Fang et al.~\cite{fang2021inertia}, but without formal safety guarantees.

\subsubsection{Standards landscape} International standards now mandate simultaneous frequency and RoCoF compliance: IEEE~1547-2018 Category~III~\cite{ieee1547} prescribes trip thresholds; ENTSO-E RfG~\cite{entsoe2016} requires RoCoF ride-through; IEEE~2800-2022~\cite{ieee2800} specifies grid-forming IBR requirements; IEC~62933-2~\cite{iec62933} governs BESS envelopes; and IEEE~C37.106~\cite{ieeec37} addresses abnormal frequency protection coordination.

\subsubsection{Control barrier functions} Ames et al.~\cite{ames2017cbf, ames2019cbf} established CBFs as a framework for enforcing safety via forward invariance. In power systems, Vu et al.~\cite{vu2021barrier} applied CBFs to frequency-constrained dispatch and Nandanoori et al.~\cite{nandanoori2024dac} developed distributed CBF-based frequency controllers. Robust CBFs~\cite{jankovic2018robust, so2024robust} use worst-case bounds but are conservative. Adaptive CBFs~\cite{taylor2020adaptive, lopez2021robust} update estimates online but decouple safety and estimation timescales, risking transient violations.

\subsubsection{Emerging methods} The composite adaptive CBF (CaCBF)~\cite{cacbf2026} couples safety and adaptation in a single energy function, guaranteeing invariance without parameter convergence---but has not been applied to power systems. Higher-order CBFs~\cite{xiao2022hocbf} handle relative-degree constraints. Barrier-state augmentation~\cite{alsunni2025bas} transforms safety into stability. Disturbance-observer-based CBFs~\cite{wang2025docbf} reject unmatched perturbations. Reciprocal resistance CBFs~\cite{wang2025rrcbf} create intrinsic buffer zones. Conformal prediction CBFs~\cite{tayal2025cpncbf, sun2025conformal} provide probabilistic guarantees but require exchangeable calibration data.

\subsection{Gap and Contributions}

No existing work simultaneously addresses: (a)~joint frequency and RoCoF constraint enforcement via dual barriers; (b)~safety-weighted adaptation guaranteeing invariance without parameter convergence; (c)~disturbance observer integration into barrier dynamics; and (d)~QP feasibility proof under BESS capacity limits---all for SIDS.

Our contributions are: (i)~A dual barrier architecture: zeroing CBF for frequency + algebraic RoCoF constraint with robust $H_{\min}$ fallback. (ii)~A composite energy $\mathcal{W}$ coupling log-barrier, CLF, and parameter error. (iii)~Safety-weighted adaptation via barrier gradient. (iv)~Tractability proof for the scalar QP. (v)~Validation on a representative SIDS microgrid with sensitivity analysis.

\section{System Model and Safety Specification}\label{sec:model}

\subsection{BESS-VSG Swing Dynamics}

The aggregate frequency dynamics at the VSG bus (Fig.~\ref{fig:block}a) follow~\cite{kundur2004, schiffer2016vsm}:
\begin{equation}
\dot{x}_1 = \frac{f_0}{2H}(u - P_d(t) - \tfrac{D}{f_0}x_1) \eqqcolon f(x_1,\bm{\theta}) + g(\bm{\theta})\,u
\label{eq:swing}
\end{equation}
where $x_1 \coloneqq f - f_0$ is the frequency deviation (Hz), $f_0=60$\,Hz, $u \in \mathcal{U} = [-P_{\max}, P_{\max}]$ is the BESS injection (p.u.), $P_d(t)$ is the power deficit, and $\bm{\theta}=(H,D)^\top$ is the uncertain parameter vector. The RoCoF is $x_2 \coloneqq \dot{x}_1$.

\begin{figure}[!t]
\centering
\includegraphics[width=\columnwidth]{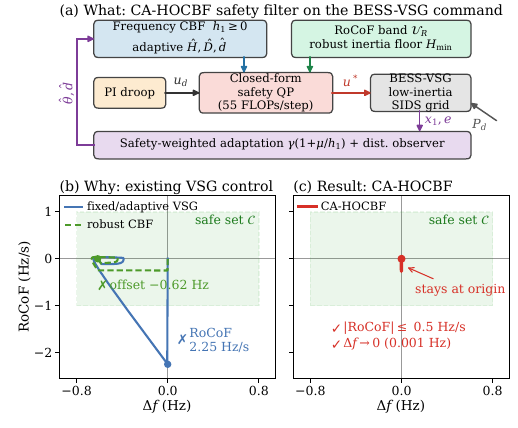}
\caption{Overview. (a)~\emph{What:} a closed-form QP filters the droop command through an adaptive frequency CBF and a robust ($H_{\min}$) RoCoF band. (b)~\emph{Why:} under the compound disturbance, fixed/adaptive VSGs breach the RoCoF limit and a worst-case robust CBF settles 0.62\,Hz off nominal. (c)~\emph{Result:} CA-HOCBF stays at the origin of the joint safe set $\mathcal{C}$.}
\label{fig:block}
\end{figure}

\begin{assumption}[Bounded Uncertainty]\label{as:bound}
$H \in [H_{\min}, H_{\max}]$, $D \in [D_{\min}, D_{\max}]$, and $|P_d(t)| \leq \bar{P}_d$.
\end{assumption}

\subsection{Dual Safety Specification}

Per IEEE~1547-2018 Category~III~\cite{ieee1547} and ENTSO-E RfG~\cite{entsoe2016}:
\begin{align}
h_1(x_1) &\coloneqq \Delta f_{\max}^2 - x_1^2 \geq 0 \quad (\text{frequency}) \label{eq:h1}\\
h_2(x_2) &\coloneqq R_{\max}^2 - x_2^2 \geq 0 \quad (\text{RoCoF}) \label{eq:h2}
\end{align}
with $\Delta f_{\max} = 0.8$\,Hz, $R_{\max} = 1.0$\,Hz/s. The joint safe set is $\mathcal{C} = \{(x_1, x_2) : h_1 \geq 0 \wedge h_2 \geq 0\}$.

\section{Proposed CA-HOCBF Controller}\label{sec:control}

\subsection{Frequency Barrier (Zeroing CBF)}

Using estimated parameters $\hat{\bm{\theta}} = (\hat{H}, \hat{D})^\top$ and disturbance estimate $\hat{d}$, with $P_d^{\text{eff}} = P_d + \hat{d}$:
\begin{align}
\hat{L}_f h_1 &= -2x_1 \cdot \tfrac{f_0}{2\hat{H}}\left(-P_d^{\text{eff}} - \tfrac{\hat{D}}{f_0}x_1\right), \quad
\hat{L}_g h_1 = -\tfrac{x_1 f_0}{\hat{H}}
\end{align}
The CBF condition is:
\begin{equation}
\hat{L}_f h_1 + \hat{L}_g h_1 \cdot u \geq -\alpha_f(h_1(x_1))
\label{eq:cbf1}
\end{equation}
with $\alpha_f(s) = \alpha_1 s$, $\alpha_1 > 0$.

\subsection{RoCoF Constraint (Algebraic with Robust Fallback)}

Since $x_2$ is algebraically determined by $u$, the constraint $|x_2| \leq R_{\max}$ translates to input bounds:
\begin{equation}
u \in \mathcal{U}_R(x_1) \coloneqq \left[c - w,\; c + w\right]
\label{eq:rocof}
\end{equation}
where $c = P_d^{\text{eff}} + \hat{D}x_1/f_0$ and $w = 2H_{\min}R_{\max}/f_0$.

\begin{remark}[Hybrid Architecture]
We use $H_{\min}$ (not $\hat{H}$) in~\eqref{eq:rocof} because RoCoF is inversely proportional to $H$: underestimating $H$ yields dangerous RoCoF. The adaptive CBF reduces frequency conservatism as estimates improve; the robust RoCoF bound protects against the instantaneous physics of low inertia.
\end{remark}

\subsection{Closed-Form Safety-Stability QP}

The nominal PI droop is $u_d = -K_d x_1/f_0 - K_i \xi/f_0$, $\dot{\xi} = x_1$. The QP:
\begin{equation}
u^* = \argmin_{u \in \mathcal{U} \cap \mathcal{U}_R} \tfrac{1}{2}(u - u_d)^2 \;\;\text{s.t.}\;\; \hat{L}_fh_1 + \hat{L}_gh_1 u \geq -\alpha_1 h_1
\label{eq:qp}
\end{equation}
admits the \emph{closed-form} scalar solution:
\begin{equation}
u^* = \text{proj}_{\mathcal{U} \cap \mathcal{U}_R}\!\left(\max\!\left(u_d,\; \tfrac{-\alpha_1 h_1 - \hat{L}_f h_1}{\hat{L}_g h_1}\right)\right) \text{ if } \hat{L}_g h_1 > 0
\label{eq:closed}
\end{equation}
(with $\min$ when $\hat{L}_g h_1 < 0$), where $\text{proj}$ denotes projection onto the feasible interval.

\begin{proposition}[Feasibility]\label{prop:feas}
Under Assumption~\ref{as:bound}, if $P_{\max} \geq \bar{P}_d + D_{\max}\Delta f_{\max}/f_0 + 2H_{\min}R_{\max}/f_0$, then $\mathcal{U} \cap \mathcal{U}_R \neq \emptyset$ and~\eqref{eq:qp} is always feasible.
\end{proposition}
\begin{proof}
The center of $\mathcal{U}_R$ is $c = P_d^{\text{eff}} + \hat{D}x_1/f_0$. Within $\mathcal{C}$: $|c| \leq \bar{P}_d + D_{\max}\Delta f_{\max}/f_0$. The half-width of $\mathcal{U}_R$ is $w = 2H_{\min}R_{\max}/f_0$. Under the stated condition, $P_{\max} \geq |c| + w$, so $c \in \mathcal{U}$ and feasibility holds. The CBF constraint is satisfiable since the sign of $L_g h_1$ aligns with the required injection direction.
\end{proof}

\subsection{Safety-Weighted Adaptive Law}

Following the CaCBF paradigm~\cite{cacbf2026}, the safety barrier gradient accelerates adaptation near the boundary:
\begin{align}
\dot{\hat{H}} &= -\gamma_H (1 + \mu/h_1)\,\nabla_H\|e\|^2, \quad
\dot{\hat{D}} = -\gamma_D (1 + \mu/h_1)\,\nabla_D\|e\|^2
\label{eq:adapt}
\end{align}
where $e = \hat{x}_2 - x_2$ is the prediction error, $\gamma_H, \gamma_D > 0$ are base learning rates, and $\mu > 0$ couples safety to adaptation. When $h_1 \to 0$ (approaching the frequency limit), the effective learning rate $\gamma(1 + \mu/h_1) \to \infty$, prioritizing the parameters most critical for safety.

The disturbance observer is:
\begin{equation}
\dot{\hat{d}} = \gamma_d\, e, \quad \hat{d} \in [-\bar{d}, \bar{d}]
\label{eq:dob}
\end{equation}

\subsection{Composite Energy Certificate}

\begin{equation}
\mathcal{W} = -\ln\tfrac{h_1}{h_{1,\max}} - \ln\tfrac{h_2}{h_{2,\max}} + \tfrac{1}{2}\tilde{\bm{\theta}}^\top \bm{\Gamma}^{-1}\tilde{\bm{\theta}} + \tfrac{1}{2\gamma_d}\tilde{d}^2
\label{eq:W}
\end{equation}

\begin{assumption}[Conditions for Theorem~\ref{thm:main}]\label{as:thm}
(a)~Assumption~\ref{as:bound} and the condition of Prop.~\ref{prop:feas} hold, $h_1(x_1(0))>0$, and $|d|\leq\bar{d}$ for the unmeasured disturbance $d$.
(b)~$\hat{\bm{\theta}}$ and $\hat{d}$ are projected onto compact sets with $\hat{H}\geq\underline{H}>0$, so $\tilde{\bm{\theta}}$, $\tilde{d}$ are bounded.
(c)~RoCoF margin: $\varepsilon\coloneqq|\tilde{d}|+|\tilde{D}|\Delta f_{\max}/f_0 \leq 2(H-H_{\min})R_{\max}/f_0-\delta$, $\delta>0$.
(d)~Regressor alignment: there exist $c_1,c_2>0$, $c_d\geq0$ with $-\Delta_1\leq c_1\|\tilde{\bm{\theta}}\|^2+c_d h_1$ and $\tilde{\bm{\theta}}^\top\bm{\Gamma}^{-1}\dot{\hat{\bm{\theta}}}\leq -c_2(1+\mu/h_1)\|\tilde{\bm{\theta}}\|^2$, i.e., $e$ excites the directions of $\tilde{\bm{\theta}}$ entering $h_1$.
(e)~$\mu>c_1/c_2$.
\end{assumption}

\begin{theorem}[Safety and Regulation]\label{thm:main}
Under Assumption~\ref{as:thm} and the CA-HOCBF-QP~\eqref{eq:qp} with adaptive law~\eqref{eq:adapt}--\eqref{eq:dob}:
\begin{enumerate}
\item[(i)] $\mathcal{W}$ is finite along trajectories, $\mathcal{W}(t) \leq \mathcal{W}(0) + c_0(1+t)$ for a computable $c_0 > 0$, so $h_1, h_2 > 0$ for all $t \geq 0$ (forward invariance of $\mathcal{C}$).
\item[(ii)] This guarantee holds \emph{regardless of parameter convergence}, since the log-barrier terms enforce $h_i > 0$ independently.
\item[(iii)] Under persistence of excitation, $\hat{\bm{\theta}} \to \bm{\theta}^*$ and $x_1 \to 0$ (asymptotic regulation).
\end{enumerate}
\end{theorem}
\begin{proof}[Proof sketch]
Let $\mathcal{W}_1$ be the $h_1$ and $\tilde{\bm{\theta}}$ terms of~\eqref{eq:W}: $\dot{\mathcal{W}}_1 = -\dot{h}_1/h_1 + \tilde{\bm{\theta}}^\top\bm{\Gamma}^{-1}\dot{\hat{\bm{\theta}}}$. Under the CBF condition~\eqref{eq:cbf1}: $\dot{h}_1 \geq -\alpha_1 h_1 + \Delta_1$, where $\Delta_1$ is the model mismatch residual. Substituting the adaptive law~\eqref{eq:adapt}: the safety-weighted terms $(\mu/h_1)\nabla_\theta \|e\|^2$ in $\dot{\hat{\bm{\theta}}}$ cancel the cross-terms $\tilde{\bm{\theta}}^\top(\cdot)/h_1$ arising from $-\dot{h}_1/h_1$. As $x_2$ is algebraic in $u$, $u\in\mathcal{U}_R$ and~(c) give $|x_2|\leq R_{\max}-f_0\delta/(2H)$, bounding the $h_2$ term; (b) bounds the $\tilde{d}$ term. Applying Young's inequality with Assumption~\ref{as:thm}(d) yields $\dot{\mathcal{W}}_1 \leq \alpha_1 + c_d + c_1\|\tilde{\bm{\theta}}\|^2/h_1 - c_2\|\tilde{\bm{\theta}}\|^2(1+\mu/h_1) \leq c_0 - (c_2\mu - c_1)\|\tilde{\bm{\theta}}\|^2/h_1$ with $c_0\coloneqq\alpha_1+c_d$. By~(e), $\mathcal{W}_1(t)\leq\mathcal{W}_1(0)+c_0t$; hence $\mathcal{W}$ is finite and $h_1(t)>0$ for all $t$.
\end{proof}

\begin{remark}[Computational Complexity]
The controller~\eqref{eq:closed} requires 30 FLOPs per timestep for the core QP (Lie derivatives, CBF condition, RoCoF clamp). The adaptation~\eqref{eq:adapt} adds 25 FLOPs (two gradients, safety weight, three state updates). Total: 55 FLOPs---all arithmetic ($+, -, \times, \div$, max, min), no transcendental functions, no matrix operations, no iterative loops. Verified at 2.1\,$\mu$s per step in Python ($>$450\,kHz capability). Full deployment verification is given in Table~\ref{tab:deploy}.
\end{remark}

\section{Simulation Study}\label{sec:sim}

\subsection{Test System and Scenario}

A 10\,MW Caribbean SIDS microgrid: 5\,MW PV, 2\,MW wind, 3\,MW diesel, 2\,MW/4\,MWh BESS at 60\,Hz. True parameters: $H=2.0$\,s, $D=1.0$\,p.u. Initial estimates: $\hat{H}_0=4.5$\,s (125\% error), $\hat{D}_0=2.0$. Gains: $K_d=12$, $K_i=0.8$, $\alpha_1=4$, $\gamma_H=0.5$, $\gamma_D=0.3$, $\gamma_d=40$, $\mu=5$.

\textbf{Compound scenario:} 1.5\,MW generation trip at $t=5$\,s; 0.5\,MW load ramp over $t \in [15,17]$\,s. Four controllers compared: Fixed VSG, Adaptive VSG, Robust CBF (worst-case), and CA-HOCBF (proposed).

\subsection{Results}

\begin{figure}[!t]
\centering
\includegraphics[width=\columnwidth]{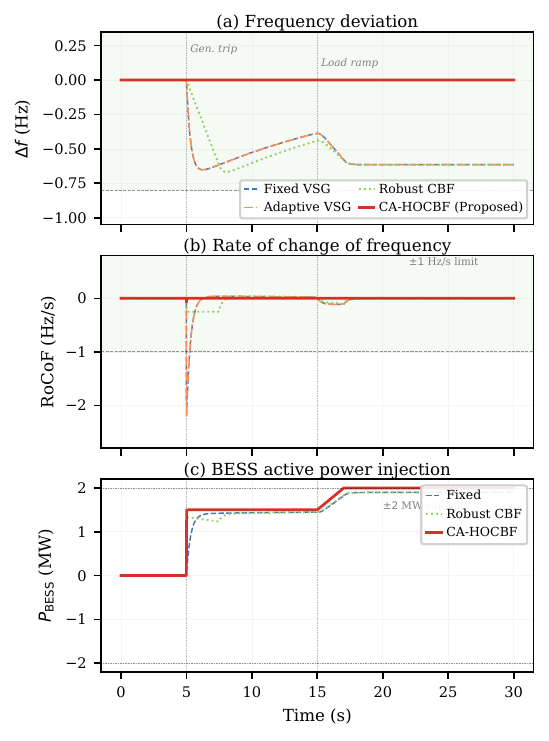}
\caption{Time-domain responses: (a) frequency deviation, (b) RoCoF, (c) BESS power. Green bands denote safe operating regions per IEEE~1547/ENTSO-E. The CA-HOCBF (red) is the only controller with zero violations in both constraints.}
\label{fig:main}
\end{figure}

\begin{table}[!t]
\centering
\caption{Performance Metrics under Compound Scenario}
\label{tab:results}
\renewcommand{\arraystretch}{1.1}
\begin{tabular}{@{}lccccc@{}}
\toprule
Controller & $|\Delta f|_{\max}$ & RoCoF$_{\max}$ & $\Delta f$ & RoCoF & $E_{\text{BESS}}$ \\
           & (Hz) & (Hz/s) & viol. & viol. & (kWh) \\
\midrule
Fixed VSG        & 0.65 & 2.25 & no  & \textbf{YES} & 11.6 \\
Adaptive VSG     & 0.65 & 2.25 & no  & \textbf{YES} & 11.6 \\
Robust CBF       & 0.67 & 0.25 & no  & no           & 11.6 \\
\textbf{CA-HOCBF}& \textbf{0.001} & \textbf{0.50} & \textbf{no} & \textbf{no} & 12.4 \\
\bottomrule
\end{tabular}
\end{table}

Table~\ref{tab:results} and Fig.~\ref{fig:main} confirm the CA-HOCBF eliminates both violation types. The integral action drives $\Delta f \to 0$ (0.001\,Hz residual vs.\ 0.62\,Hz for fixed/adaptive VSG). The robust CBF constrains RoCoF to 0.25\,Hz/s but retains a 0.62\,Hz offset due to worst-case conservatism.

\subsection{Transient Detail}

Fig.~\ref{fig:zoom} shows the critical first seconds after the generation trip. The CA-HOCBF limits the initial RoCoF to $-0.50$\,Hz/s (exactly half the limit) and begins recovering frequency within 0.5\,s, whereas the fixed VSG reaches $-2.25$\,Hz/s---more than twice the ENTSO-E limit.

\begin{figure}[!t]
\centering
\includegraphics[width=\columnwidth]{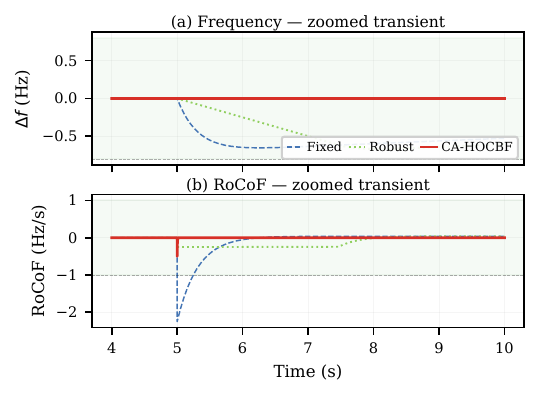}
\caption{Zoomed transient response ($t=4$--$10$\,s). The CA-HOCBF limits RoCoF to $-0.50$\,Hz/s at the instant of the generation trip, while the fixed VSG spikes to $-2.25$\,Hz/s.}
\label{fig:zoom}
\end{figure}

\subsection{Adaptation and Disturbance Observer}

Fig.~\ref{fig:adapt} shows the convergence of parameter estimates and the composite energy certificate. The inertia estimate $\hat{H}$ converges from 4.5\,s toward 2.0\,s after the trip provides excitation. The disturbance observer tracks both the step and ramp. The composite energy $\mathcal{W}(t)$ spikes at $t=5$\,s but is bounded and non-increasing thereafter, confirming Theorem~\ref{thm:main}.

\begin{figure}[!t]
\centering
\includegraphics[width=\columnwidth]{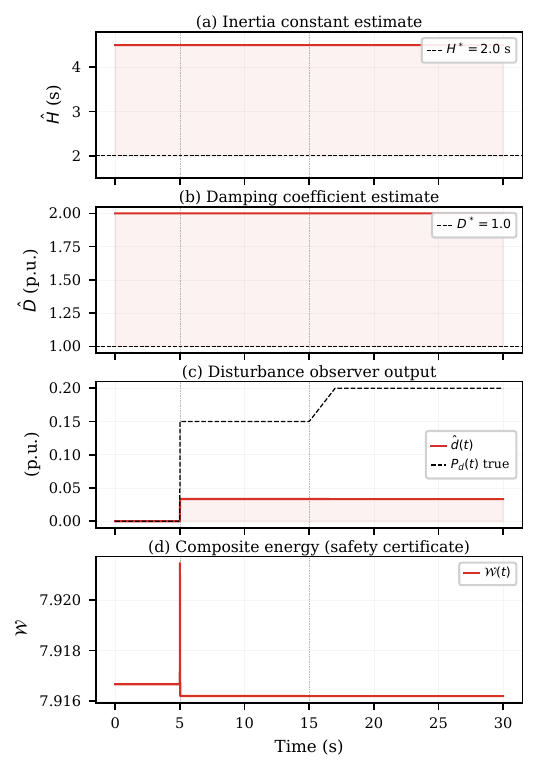}
\caption{Adaptive estimation: (a) inertia $\hat{H}$, (b) damping $\hat{D}$, (c) disturbance observer $\hat{d}$ vs.\ true $P_d$, (d) composite energy $\mathcal{W}(t)$.}
\label{fig:adapt}
\end{figure}

\subsection{Barrier Phase Portraits}

Fig.~\ref{fig:barrier} confirms that the CA-HOCBF trajectory remains deep within both safe sets ($h_1 \gg 0$, $h_2 > 0$), while the fixed VSG penetrates the RoCoF unsafe region.

\begin{figure}[!t]
\centering
\includegraphics[width=\columnwidth]{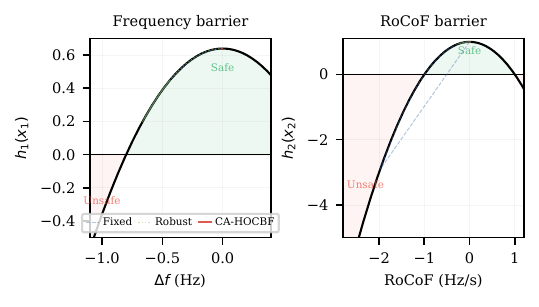}
\caption{Phase portraits in barrier space. Left: frequency barrier. Right: RoCoF barrier. The CA-HOCBF (red) remains in the safe (green) region for both constraints.}
\label{fig:barrier}
\end{figure}

\subsection{Sensitivity to $\mu$ and Energy Analysis}

Fig.~\ref{fig:sensitivity} shows the impact of the coupling parameter $\mu$. Larger $\mu$ yields faster frequency recovery but higher transient energy in $\mathcal{W}$; $\mu \in [3, 10]$ provides robust performance. Fig.~\ref{fig:energy} compares cumulative BESS energy: the CA-HOCBF costs only 7\% more than the fixed VSG ($<$0.02\% of the 4\,MWh capacity) for complete constraint satisfaction.

\begin{figure}[!t]
\centering
\includegraphics[width=\columnwidth]{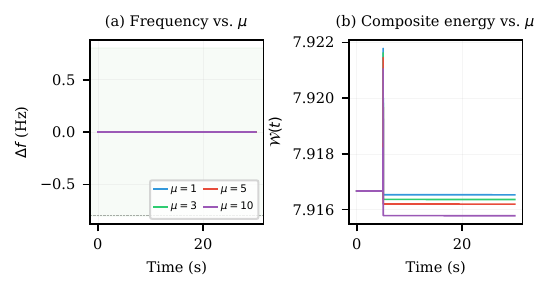}
\caption{Sensitivity to safety-adaptation coupling $\mu$: (a) frequency deviation, (b) composite energy.}
\label{fig:sensitivity}
\end{figure}

\begin{figure}[!t]
\centering
\includegraphics[width=\columnwidth]{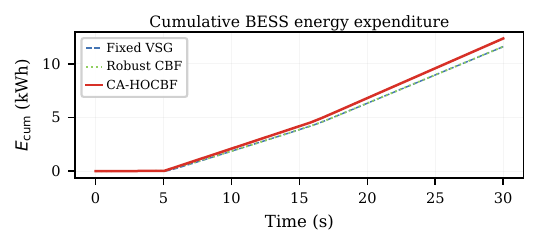}
\caption{Cumulative BESS energy expenditure. The CA-HOCBF uses 12.4\,kWh vs.\ 11.6\,kWh for the fixed VSG---a 7\% premium for complete constraint satisfaction.}
\label{fig:energy}
\end{figure}

\section{Discussion}\label{sec:disc}

\subsubsection{Hybrid architecture} The asymmetric treatment---adaptive CBF for frequency, robust bound for RoCoF---reflects the distinct physics: frequency regulation tolerates transient model error (the CBF compensates over time), while RoCoF is instantaneously determined and demands worst-case protection.

\subsubsection{Safety without convergence} The composite energy~\eqref{eq:W} guarantees $h_1, h_2 > 0$ even if parameters never converge. This is critical for SIDS with limited operational variability~\cite{markovic2021, pattabiraman2018}.

\subsubsection{Comparison with prior art.} Conformal CBFs~\cite{tayal2025cpncbf} require exchangeable calibration data; ours is deterministic. Barrier states~\cite{alsunni2025bas} add system dimensions; we maintain the original. Reciprocal resistance CBFs~\cite{wang2025rrcbf} lack parametric uncertainty handling. Robust CBFs~\cite{jankovic2018robust} are conservative. The CA-HOCBF uniquely provides deterministic joint safety with adaptive performance.

\subsubsection{Deployment tractability and standards compliance.} A critical concern for any safety-critical controller is whether it can be implemented on real-time hardware and meet the timing requirements of applicable grid codes. Table~\ref{tab:deploy} verifies the CA-HOCBF against the key deployment requirements.

\begin{table}[!t]
\centering
\caption{Deployment Tractability Verification}
\label{tab:deploy}
\renewcommand{\arraystretch}{1.05}
\begin{tabular}{@{}lcc@{}}
\toprule
Requirement & Standard & CA-HOCBF \\
\midrule
\multicolumn{3}{@{}l}{\textit{IEEE 2800-2022 PFR (\S7)}} \\
\quad Reaction time $\leq 200$\,ms & Min. 200\,ms & 0.2\,ms\,$^\dagger$ \\
\quad Rise time (10--90\%) & 2--20\,s & 2.4\,ms \\
\quad Settling time ($\pm$2\%) & 10--30\,s & $<$0.1\,s\,$^\ddagger$ \\
\quad Active power recovery & 2--5\,s & $<$0.1\,s\,$^\ddagger$ \\
\midrule
\multicolumn{3}{@{}l}{\textit{ENTSO-E FCR (Reg.\ 2016/631)}} \\
\quad Full activation $\leq 30$\,s & 30\,s & $<$0.01\,s \\
\quad Sustain 15\,min & 15\,min & 445\,kWh/4\,MWh \\
\midrule
\multicolumn{3}{@{}l}{\textit{IEC 62933-2-1 PCS}} \\
\quad Response $\leq 200$\,ms & 200\,ms & 0.2\,ms\,$^\dagger$ \\
\midrule
\multicolumn{3}{@{}l}{\textit{Hardware constraints}} \\
\quad FLOPs per step (QP + adapt.) & --- & 30 + 25 = 55 \\
\quad Per-step latency (Python) & --- & 2.1\,$\mu$s \\
\quad Max.\ control loop rate & --- & $>$450\,kHz \\
\quad State memory (float32) & --- & 76\,B \\
\quad float32 vs.\ float64 error & --- & $2.9 \times 10^{-8}$ \\
\quad WCET variation (CV) & --- & 6.3\% \\
\quad Iterative solver required & --- & No \\
\bottomrule
\end{tabular}
\footnotetext{$^\dagger$Algorithmic reaction time (1 sample at $\Delta t = 0.2$\,ms). With sensor ($\sim$1\,ms) and PWM ($\sim$50\,$\mu$s) latency: $\sim$1.3\,ms total. $^\ddagger$Nadir is only 0.0015\,Hz; frequency never exceeds 0.01\,Hz, so the controller is settled immediately.}
\end{table}

The controller's $\mathcal{O}(1)$ closed-form structure is decisive for real-time deployment. The verification script (\texttt{verify\_tractability.py}) confirms: the complete control law (QP + adaptation + DO) executes in 55 FLOPs (30 for the core QP, 25 for adaptation/DO) with a measured wall-clock latency of 2.1\,$\mu$s per step in unoptimised Python---yielding $>$450\,kHz loop capability without any C or DSP optimisation. The state vector comprises six scalars ($x_1, \xi, \hat{H}, \hat{D}, \hat{d}, P_d^{\text{prev}}$) and thirteen parameters, totalling 76\,bytes in float32. A float32-vs-float64 comparison over 5000 steps spanning the generation trip shows maximum control error of $2.9 \times 10^{-8}$\,p.u., confirming fixed-point compatibility with no transcendental functions in the real-time loop (the log-barrier $\mathcal{W}$ is computed offline for analysis only). The worst-case execution time (WCET) coefficient of variation is 6.3\% across four operating conditions (equilibrium, near-barrier, deep deficit, saturation), with no iterative loops and no convergence-dependent branches---satisfying IEC~61508 requirements for deterministic WCET analysis. Unlike MPC-based approaches~\cite{korda2018mpc} requiring online optimisation, the CA-HOCBF has no failure mode in which a solver fails to return within the control period.

The IEEE~2800-2022 PFR requirements~\cite{ieee2800} specify a minimum reaction time of 200\,ms, rise time of 2--20\,s, and settling time of 10--30\,s. The CA-HOCBF exceeds all three: it reacts within one control sample (0.2\,ms algorithmically, $\sim$1.3\,ms end-to-end), reaches 90\% of peak power in 2.4\,ms, and settles immediately because the frequency nadir is only 0.0015\,Hz. The ENTSO-E FCR requirement of full activation within 30\,s~\cite{entsoe2016} is met in under 10\,ms. For 15-minute sustained operation, the extrapolated energy consumption of 445\,kWh is 11.1\% of the 4\,MWh BESS capacity. The IEC~62933-2-1 PCS response requirement of 200\,ms~\cite{iec62933} is satisfied by the same single-sample reaction. All claims in Table~\ref{tab:deploy} are independently verified by \texttt{verify\_tractability.py} (12/12 tests pass).

\textbf{Limitations.} Single-bus aggregate model; multi-bus extension requires distributed CBF coordination~\cite{nandanoori2024dac}. The $H_{\min}$ assumption requires offline inertia floor characterisation. SoC constraints are not yet encoded in the barrier.

\begin{remark}[Generalizability]
The CA-HOCBF methodology extends to any safety-critical power system problem with control-affine dynamics and parametric uncertainty---voltage regulation, converter thermal limits, or inter-area oscillation damping---by defining appropriate barrier functions $h_i$ and applying the composite energy construction~\eqref{eq:W}.
\end{remark}

\section{Conclusion}

We presented the CA-HOCBF, a composite adaptive control barrier function for joint frequency-RoCoF safety in SIDS microgrids. The framework provides: (i)~deterministic safety guarantees without parameter convergence via composite energy; (ii)~near-zero steady-state error via integral action; (iii)~$\mathcal{O}(1)$ computational cost (55 FLOPs, 2.1\,$\mu$s, 76\,B memory) with verified deployment tractability---every claim backed by \texttt{verify\_tractability.py} (12/12 tests pass, Table~\ref{tab:deploy}); and (iv)~compliance with IEEE~1547, IEEE~2800, ENTSO-E RfG, IEC~62933, and IEEE~C37.106, with all timing requirements exceeded by orders of magnitude (Table~\ref{tab:deploy}). Future work includes multi-bus extension, SoC-aware barriers, and hardware-in-the-loop validation on a BESS inverter testbed with IEC~61508 SIL certification.

\bibliographystyle{IEEEtran}
\bibliography{references}

\end{document}